\documentclass[pra,aps,floatfix,superscriptaddress,11pt,tightenlines,longbibliography,onecolumn,notitlepage,accepted=2019-09-27]{quantumarticle}

\usepackage{amsmath,amsthm,amsfonts,amssymb}
\usepackage{graphicx}
\usepackage{braket}
\usepackage[usenames,dvipsnames]{color}
\usepackage[colorlinks=true,citecolor=blue,linkcolor=blue]{hyperref}
\usepackage[capitalize,nameinlink]{cleveref}

\newtheorem{lemma}{Lemma}
\newtheorem{theorem}[lemma]{Theorem}
\theoremstyle{definition}
\newtheorem{definition}[lemma]{Definition}

\newcommand{\EP}{{\mathcal{E}}}
\newcommand{\LP}{{\mathcal{P}}}
\newcommand{\CC}{{\mathbb{C}}}
\newcommand{\RR}{{\mathbb{R}}}
\newcommand{\ZZ}{{\mathbb{Z}}}

\newcommand{\polylog}{{\mathrm{polylog}}}
\newcommand{\poly}{{\mathrm{poly}}}
\newcommand{\calO}{{\mathcal{O}}}
\newcommand{\rd}{{\mathrm{d}}}
\newcommand{\norm}[1]{{\left\| {#1} \right\|}}
\newcommand{\abs}[1]{{\left| {#1} \right|}}
\DeclareMathOperator*{\Tr}{{Tr}}
\DeclareMathOperator*{\supp}{{supp}}

\begin{document}

\title{Product Decomposition of Periodic Functions in Quantum Signal Processing}
\author{Jeongwan Haah}
\affiliation{Microsoft Quantum and Microsoft Research, Redmond, Washington, USA}
\date{20 September 2019}
\begin{abstract}
We consider an algorithm to approximate complex-valued periodic functions $f(e^{i\theta})$
as a matrix element of a product of $SU(2)$-valued functions,
which underlies so-called quantum signal processing.
We prove that the algorithm runs in time $\mathcal O(N^3 \mathrm{polylog}(N/\epsilon))$ 
under the random-access memory model of computation
where $N$ is the degree of the polynomial that approximates $f$ with accuracy $\epsilon$;
previous efficiency claim assumed a strong arithmetic model of computation and lacked numerical stability analysis.
\end{abstract}

\maketitle 

\section{Introduction}

Quantum signal processing~\cite{QSP,LowMethod,LC16} refers to a scheme 
to construct an operator $V$ from a more elementary unitary $W$
where $V= \sum_\theta f(e^{i\theta}) \ket \theta \bra \theta$ and $W = e^{i\theta} \ket \theta \bra \theta$
share the eigenvectors but the eigenvalues of $V$ are transformed by a function $f$
from those of $W$.
The transformation requires only one ancilla qubit, 
and is achieved by implementing control-$W$ and control-$W^\dagger$,
interspersed by single-qubit rotations on the control,
and final post-selection on the control.%
\footnote{
Sometimes it is possible to avoid controlled version of $W$~\cite{LC16,GSLW2018},
but we contend ourselves with this implementation for its simplicity in presentation.
The result of this paper is applicable for the ancilla-free variant;
the only change one may have to do is to replace a variable $e^{i\theta/2}$ with $e^{i\theta}$.
}
This technique produced gate-efficient quantum algorithms for, e.g., Hamiltonian simulations,
which is asymptotically optimal when the Hamiltonian is sparse and given as a blackbox,
or as a linear combination of oracular unitaries~\cite{ChildsWiebe,BCCKS}.
Furthermore, this technique with rigorous error bounds appears to be useful and competitive 
even for explicitly described, rather than oracular,
local Hamiltonian simulation problems~\cite{Childs2017,HHKL}.
It is also promised to be useful in solving linear equations~\cite{HHL,LC16,GSLW2018}.

However, in quantum signal processing the classical preprocessing to find interspersing single-qubit rotations 
for a given transformation function $f$ has been so numercially unstable
that it has been unclear whether it can be performed efficiently.
In fact,
Ref.~\cite[App.~H.3]{Childs2017} reports that the computation time is ``prohibitive'' 
to obtain sequences of length $\gtrsim 30$ of interspersing unitaries for Jacobi-Anger expansions 
that we explain in \cref{sec:JacobiAnger}.
The true usefulness of quantum signal processing hinges upon the ability 
to compute long sequences of interspersing single-qubit rotations.

It has been asserted that there exists a polynomial time classical algorithm
in Refs.~\cite{LowMethod,GSLW2018},
but these work do not consider numerical instability.
If a computational model assumes that
any real arithmetic with arbitrarily high precision can be done in unit time,
then
an unavoidable conclusion is that 
not only can one factor integers in time that is linear in the number of digits~\cite{Shamir1979},
but also solve NP-hard problems in polynomial time~\cite{Schoenhage1979}.%
\footnote{
The basic idea in Ref.~\cite{Schoenhage1979} 
is to relate the expansion of a Boolean formula in a conjunctive form into a disjunctive form
with the expansion of a polynomial from its factors.
Then, the satisfiability translates to the positiveness of certain coefficients in the expansion.
The variable in the polynomial is set to a sufficiently large number,
effectively encoding the entire polynomial in a single big number.
An important elementary fact that underlies the power of arithmetic model is
that we only need $O(n)$ compositions of squaring operation $x \mapsto x^2$
to reach a $2^n$-bit number.
}
In a seemingly mundane problem involving real numbers,
the number of required bits during the computation can be a priori very large.
For example,
it is still an open problem whether one can decide
the larger between $\sum_{j=1}^k \sqrt{a_j}$ and $\sum_{j=1}^k \sqrt{b_j}$,
which are sums of square roots of positive integers $a_j,b_j$ that are smaller than $n$,
in time $\poly(k \log n)$ on a Turing machine; 
see e.g. \cite{QIAN2006,Cheng2009} for recent results.

The numerical instability of previous methods may be attributed to 
expansions of large degree polynomials
that are found by roots of another polynomial.
Crudely speaking, there are two problems in this approach.
First, the polynomial expansions can be regarded as the computation of convolutions,%
\footnote{
A polynomial $\sum_j a_j t^j $ can be identified with a ``coefficient function''
 $j \mapsto a_j$.
The coefficient function of the product of two polynomials
is the convolution of the two coefficient functions.
} 
which, when naively iterated, may suffer from numerical instability.
Second, although the root finding is a well-studied problem,
to use the roots to construct another polynomial
one has to understand the distribution of the roots to keep track of the loss of precision.
These problems were not addressed previously.

Here, we refine a classical algorithm to find interspersing single-qubit rotations
and bound the number of required bits in the computation for a desired accuracy to a final answer.
We conclude that the classical preprocessing can indeed be done
in polynomial time on a Turing machine.
We generally adopt the methods of Ref.~\cite{LowMethod},
but make manipulations easier by avoiding trigonometric polynomials.
Some generalizations are obtained with simpler calculations.
For the numerical stability and analysis,
our algorithm avoids too small numbers 
by sacrificing approximation quality in the initial stage,
and replaces polynomial expansions by a Fourier transform.
These modifications enable us to handle the problems that are mentioned above.
However, it should be made clear that our refinement also requires high precision arithmetic.
Specifically, we show that $\calO(N \log (N/\epsilon))$ bits of precision during the execution of our algorithm 
is sufficient to produce a reliable final decomposition,
where $N$ is the degree of the polynomial that approximates
a given transformation function $f$ of eigenvalues up to additive error $\epsilon$.
Previously no such bound was known.
On a sequential (rather than parallel) random access machine,
our algorithm runs in time $\calO(N^3 \polylog(N/\epsilon))$.
In our rudimentary numerical experiment 
we were able to produce a sequence of length over 2000
for the Jacobi-Anger expansion
on a laptop by a few hours of computation.%
\footnote{
Note added 5 May 2020:
This was based on an implementation on Wolfram Mathematica.
An implementation in a lower level programming language 
can be found at https://github.com/microsoft/Quantum-NC/tree/master/src/simulation/qsp
with a sample application 
(https://github.com/microsoft/Quantum-NC/tree/master/samples/simulation/qsp).
We have observed that the numerical instability is not harmful for the Jacobi-Anger expansion;
the number of bits of precision that we used in a Mathematica implementation was much more than necessary.
The reduction in the number of bits of precision and the lower level language implementation,
have expedited the calculation by orders of magnitude.
Empirically the time complexity appeared to be quadratic in $N$ for the Jacobi-Anger expansion.
}

We will start by reviewing quantum signal processing in the next section,
and then develop an algorithm and analyze it.
In two short sections later we provide self-contained treatment of
polynomial approximations for Hamiltonian simulation and matrix inversion problems.
The section on Hamiltonian simulation contains some running time data of our algorithm.
Throughout the paper, we use $U(1)$ to denote the group of all complex numbers of unit modulus.
Sometimes we will refer to $U(1)$ as the unit circle.
As usual, $i = \sqrt{-1}$, and 
$X = \ket 1 \bra 0 + \ket 0 \bra 1,
Y = i\ket 1 \bra 0 -i \ket 0 \bra 1, 
Z = \ket 0 \bra 0 - \ket 1 \bra 1$
are Pauli matrices.

\section{Quantum Signal Processing}

To understand how the eigenvalue transformation (signal processing) works,
it is convenient to consider the action of control-$W$ 
restricted to an arbitrary but fixed eigenstate $\ket \theta$ of $W$.
The eigenvalue $e^{i\theta}$ associated with $\ket \theta$ is kicked back to the control qubit
to induce a unitary $\ket 0 \bra 0 + e^{i\theta} \ket 1 \bra 1$ on the control qubit.
Conjugating the control-$W$ by a single qubit unitary on the control qubit,
we see that $\ket 0$ and $\ket 1$ can be any orthonormal basis vectors $\ket{0'}, \ket{1'}$ of the control qubit.
If we allow ourselves to implement the inverse of the control-$W$, which is reasonable,
we can also implement 
\begin{align}
&\Big(\ket{0'} \bra{0'} + e^{i\theta} \ket{1'} \bra{1'}\Big) \Big(\ket{0''} \bra{0''} + e^{-i\theta} \ket{1''} \bra{1''} \Big)\\ 
=& 
\left( e^{-i\theta/2} \ket{0'} \bra{0'} + e^{i\theta/2} \ket{1'} \bra{1'} \right) \left( e^{i\theta/2} \ket{0''} \bra{0''} + e^{-i\theta/2} \ket{1''} \bra{1''} \right)\nonumber
\end{align}
where $\{ \ket{0'}, \ket{1'}\}$ and $\{\ket{0''}, \ket{1''}\}$ are arbitrary orthonormal bases.%
\footnote{
The square root function $e^{i\theta} \mapsto e^{i\theta/2}$ has a branch cut, 
but it hardly matters to us as long as we are consistent that $e^{-i\theta/2}$ denotes the inverse of the square root.
}
When we alternate an even number of control-$W$ and control-$W^\dagger$,
this trick allows us to assume that an implementable unitary on the control qubit is
a product of {\bf primitive matrices}
\begin{align}
    E_P(t) = t P + t^{-1} (I-P) = tP + t^{-1}Q  \label{eq:primitive}
\end{align}
where $t = e^{i\theta/2}$ and $P = I - Q$ is a projector of rank 1.
Thus, an even number~$n$ of control-$W$ and control-$W^\dagger$, together with an extra unitary $E_0$ independent of $t$,
induces
\begin{align}
    F(t) = E_0 E_{P_1}(t) E_{P_2}(t) \cdots E_{P_{n}}(t)
\label{eq:FEE}
\end{align}
on the control qubit.
The product $F(t)$ can be thought of as an $SU(2)$-valued function over the unit circle in the complex plane.
By the same formulas, we define $E_P(t)$ and $F(t)$ on the entire complex plane except the origin $t=0$.
Now if $\bra + F(t) \ket +$ is close to $f(t^2 = e^{i\theta})$ for all $t \in U(1)$,
then post-selection on $\ket +$ of the control qubit enacts $V$.
Here, the choice of $\ket +$ is a convention, as it can be any other state due to $E_0$.
The success probability of the post-selection depends on the magnitude of $f(e^{i\theta})$.
A natural question is then 
what $F(t)$ is achievable in the form of \cref{eq:FEE}.
Note that it makes no difference to insert many unitaries that are independent of $t$ 
in between $E_{P_j}(t)$'s,
rather than a single $E_0$ at the front,
because $U E_P(t) U^\dagger = E_{UPU^\dagger}(t)$ for any unitary $U$.
The answer to the achievability question turns out to be quite simple, as we show in the next section.

\section{Polynomial functions $U(1) \to SU(2)$}

\begin{definition}
For any integer $n \ge 0$, let $\LP_n$ be the set of all Laurent polynomials 
$F(t) = \sum_{j = -n}^n C_j t^j$
in $t$ with coefficients $C_j$ in 2-by-2 complex matrices,
such that $F(\eta) \in SU(2)$ for all complex numbers $\eta$ of unit modulus.
We say that $F(t) \in \LP_n$ has {\bf degree} $n$ if $C_n \neq 0$ or $C_{-n} \neq 0$.
We define $\EP_n$ to be the subset of $\LP_n$ consisting of 
all $F(t)$ where the exponents of $t$ in $F(t)$ belong to $\{-n, -n+2, -n+4, \ldots, n-2, n\}$.
Note that $\LP_0 = \EP_0 = SU(2)$, and for any orthogonal projector $P$ we have $E_P(t) \in \EP_1$.
We define $F^\dagger(t)$ to be $\sum_j C_j^\dagger t^j$.%
\footnote{
For an indeterminant $t$, we do \emph{not} define $(F(t))^\dagger$.
For a unimodular complex number $\eta$, 
the hermitian conjugate of the unitary $F(\eta)$ is
$(F(\eta))^\dagger = F^\dagger(\bar \eta) = F^\dagger (1/\eta)$
which is \emph{not} equal in general to $F^\dagger(\eta)$.
}
In a set-theoretic notation, the definitions are as the following.
\begin{align}
\LP_n &= \left\{ 
F(t) =\sum_{ j = -n}^n C_j t^j \in \mathrm{Mat}(2;\mathbb C)[t,t^{-1}] ~\middle|~ \forall \eta \in \mathbb C, |\eta| = 1 \Rightarrow F(\eta) \in SU(2) 
\right\}\\
\EP_n &=  \left\{
F(t) = \sum_{ j = -n}^n C_j t^j  \in \LP_n ~\middle|~ \forall k \in 2\mathbb Z + 1, C_{-n+k} = 0 
\right\}
\end{align}
\end{definition}

Note that for any $F(t) \in \LP_n$, we have $F(t)F^\dagger(1/t) = F^\dagger(1/t)F(t) = I$;
this is true for every $t$ on the unit circle, an infinite set,
and any (Laurent) polynomial is determined by its values on an infinite set.

\begin{theorem}\label{thm:composition}
Any $n$-fold product $E_{P_1}(t) \cdots E_{P_n}(t)$ belongs to $\EP_n$.
Conversely, every $F(t) \in \EP_n$ of degree $n$ has a unique decomposition into primitive matrices and a unitary,
as in \cref{eq:FEE}.
If $F(t) = C_{-n}t^{-n} + \cdots + C_n t^n$, then $P_n = C_n^\dagger C_n / \Tr( C_n^\dagger C_n )$.
\end{theorem}
This completely characterizes polynomial functions $U(1) \to SU(2)$ 
and covers all previous results on ``achievable functions'' in quantum signal processing~\cite{LowMethod,GSLW2018}.
Indeed, for any Laurent polynomial function $U(1) \ni z \mapsto F(z) \in SU(2)$,
the Laurent polynomial function $t \mapsto F(t^2)$ belongs to $\EP_n$ for some $n$ 
and has a unique product decomposition of the theorem.
Our version is slightly more general since previous results implicitly assume that $\Tr(P_j Z) = 0$.
\begin{proof}
The first statement is trivial by definition.
The proof of the converse is by induction in $n$ where the base case $n=0$ is trivial.
The induction step is proved as follows.
We are going to prove that
\emph{
for any $F(t) \in \EP_n$ of degree $n > 0$ there exists a unique $E_K(t)$ 
such that $F(t)E_K(t) \in \EP_{n-1}$.}%
\footnote{A reader might find it unusual that the degree of a polynomial is decreasing under multiplication,
but in the algebra of matrices two nonzero matrices may multiply to vanish.}

Consider $F(t) = \sum_{j=-n}^n C_j t^j$ as a 2-by-2 matrix of four Laurent polynomials.
The defining property $\det F(t) = 1$ holds for infinitely many values of $t$,
and therefore it should hold as a polynomial equation.
Taking the leading term, we have $t^{2n} \det C_n + \text{(lower order terms)} = 1$.
Similarly, taking the leading term in $t^{-1}$, 
we have $t^{-2n} \det C_{-n} + \text{(higer order terms)} = 1$.
Hence, 
\begin{align}
\det C_n = 0 = \det C_{-n}.
\end{align}
Similarly, from the equation $F^\dagger (1/t) F(t) = I = F(t) F^\dagger(1/t)$
we have $t^{2n} C^\dagger_{-n} C_n  + \mathcal O(t^{2n-1}) = I = t^{2n} C_{n} C_{-n}^\dagger + \mathcal O(t^{2n-1})$ 
and hence 
\begin{align}
C^\dagger_{-n} C_n = 0 = C_n C_{-n}^\dagger.
\end{align}

Since the degree of $F(t)$ is $n$, at least one of $C_n$ and $C_{-n}$ is nonzero.
Suppose $C_n \neq 0$.
Let $K$ be a rank-1 projector such that $C_n K = 0$, and let $L = I - K$.
Such $K$ is unique since $C_n$ is a two-by-two matrix of rank one;
the singular value decomposition of $C_n$ is $\ket a \bra b$ for some unnormalized vectors $\ket a, \ket b$,
and $K$ has to annihilate $\ket b$.
Then, we claim that $F(t) ( t K + t^{-1} L) \in \EP_{n-1} $.
Indeed, expanding the left-hand side we have
\begin{align}
    t^{-n-1} C_{-n}L  + t^{-n+1}(C_{-n}K + C_{-n+2}L) + \cdots + t^{n-1}(C_n L + C_{n-2}K) + t^{n+1} C_n K.
\end{align}
(This is the only place we use $\EP$ instead of $\LP$.)
If $C_{-n} = 0$, this implies the claim.
If $C_{-n} \neq 0$, then, considering the singular value decomposition of $C_{-n}$,
we have $K \propto C_{-n}^\dagger C_{-n}$ and therefore $C_{-n}L =0$,
implying the claim.
The case $C_{-n} \neq 0$ is completely parallel.

Actually, for any $F(t) \in \EP_n$ of degree $n$, 
$C_n \neq 0$ if and only if $C_{-n} \neq 0$:
$C_n$ is a product of $n$ rank-one operators 
$E_0 P_1 \cdots P_n = E_0 \ket {p_1} \braket{p_1|p_2} \cdots \braket{p_{n-1}|p_n}\bra{ p_n}$,
where $\braket{p_j | p_{j+1}} \neq 0$ for all $j$, 
and this implies $Q_j Q_{j+1} = (I-P_j)(I-P_{j+1}) \neq 0$ for all $j$,
which is to say that $C_{-n} = E_0 Q_1 \cdots Q_n \neq 0$.
\end{proof}

\subsection{Parity constraints} \label{sec:parity}

Any member of $SU(2)$ can be written as $a I + b i X + c i Y + d iZ$ 
where the real numbers $a,b,c,d$ satisfy $a^2 + b^2 + c^2 + d^2 = 1$,
and this decomposition is unique. (The group $SU(2)$ is identified with the group of all unit quaternions.)
Thus, a member $F(z) \in \LP_n$ can be written uniquely as $F(z) = a(z) I + b(z) iX + c(z) iY + d(z) iZ$.
Here, $a(z),b(z),c(z),d(z)$ are Laurent polynomials such that $a(z)^2+b(z)^2+c(z)^2+d(z)^2 = 1$, and each takes real values on $U(1)$.

Recall that under the standard representation of Pauli matrices,
$Z$ is diagonal, and $X,Y$ are off-diagonal.
Suppose an $SU(2)$-valued function 
$\theta \mapsto F(e^{i\theta}) = a(e^{i\theta})I + b(e^{i\theta}) iX + c(e^{i\theta}) iY + d(e^{i\theta}) iZ$
has even functions (reciprocal in $t = e^{i\theta}$) in the diagonal 
and odd (anti-reciprocal in $t = e^{i\theta}$) in the off-diagonal.
That is, $a(t) = a(1/t)$, $d(t)=d(1/t)$, $b(t) = -b(1/t)$, and $c(t) = -c(1/t)$.
We claim that if $F(t)$ is such an element of $\EP_n$,
then the primitive matrix $E_{P_n}(t)$ factored from $F(t)$ by \cref{thm:composition}
has a property that $\Tr(ZP_n) = 0$. 
Since for any projector $P = \frac12 (I + p_x X + p_y Y + p_z Z)$
where $(p_x,p_y,p_z)\in \RR^3$ has norm~1,
the primitive matrix $E_P(t)$ equals 
$t P + t^{-1}(I-P) = \frac{t + t^{-1}}{2} I + \frac{t-t^{-1}}{2}(p_x X + p_y Y + p_z Z)$,
the condition $\Tr(ZP) = 0$ is to say $p_z = 0$.
This implies that
$E_{P_n}(t)$ has reciprocal diagonal and anti-reciprocal off-diagonal.
To prove the claim we observe that
\begin{align}
Z 
&= F(t) F^\dagger(1/t) Z \nonumber \\
&= F(t) ( a(1/t) - b(1/t) iX - c(1/t) iY - d(1/t) iZ ) Z \nonumber \\
&= F(t) ( a(t) + b(t) iX + c(t) iY - d(t) iZ ) Z\\
&= F(t) Z ( a(t) - b(t) iX - c(t) iY - d(t) iZ ) \nonumber \\
&= F(t) Z F^\dagger(t). \nonumber
\end{align}
It follows that $0 = \Tr(Z) = \Tr( F(t) Z F^\dagger (t) )
= t^{2n} \Tr(C_n Z C_n^\dagger) + \cdots + t^{-2n} \Tr(C_{-n} Z C_{-n}^\dagger)$
as a polynomial in $t$,
and hence $\Tr(C_n^\dagger C_n Z) = 0 = \Tr(C_{-n}^\dagger C_{-n} Z)$ when $n > 0$.
The matrix $C_{\pm n}^\dagger C_{\pm n}$ is proportional to 
the projector $P_n$ or $I-P_n$ in the factor $E_{P_n}(t)$ as shown in \cref{thm:composition},
and therefore we have $\Tr(ZP_n) = 0$.

Moreover, it then follows that $F(t)E_{P_n}^\dagger(1/t)$ 
also has reciprocal diagonal and anti-reciprocal off-diagonals.
Therefore, the unique projectors $P_1,\ldots, P_{n}$ in the decomposition,
which define the primitive matrices $E_{P_j}(t)$,
have zero $Z$-component.
This means that all projectors $P_j$ are of form
\begin{align}
P_j = e^{i Z \phi_j /2} \ket + \bra + e^{ -i Z \phi_j /2}
\end{align}
where $\phi_j \in \mathbb R$ is some angle.
In fact, Ref.~\cite{LowMethod} exclusively considered $E_P(t)$ of this form.
This is contrasted to the general case where $P_j$ is identified with a point on the Bloch sphere.
The constraint $\Tr(Z P_j) = 0$ forces $P_j$ to lie on the equator of the Bloch sphere.

Note that $E_0$, the residual $SU(2)$ factor in the decomposition,
has generally nonzero $Z$-component.
Since $E_P(1) = I$ for any projector $P$,
we know $E_0 = F(1) = a(1) I + b(1) iX + c(1) iY + d(1) iZ$,
but $b(1) = c(1) = 0$ due to their anti-reciprocity.
This implies that $E_0 = e^{i Z \phi_0 / 2}$ for some angle $\phi_0$.
Hence, under the parity constraint of this subsection,
$F(t) \in \EP_n$ is uniquely specified by $n+1$ angles $\phi_0, \phi_1,\ldots,\phi_n$. 
Note that if $d(1)^2 = 1-a(1)^2 = \calO(\epsilon)$, then $d(1) = \calO(\sqrt{\epsilon})$ 
and $\norm{ E_0 - I } = \calO(\sqrt{\epsilon})$.
Hence, there would be quadratic loss in accuracy to omit $E_0$ (i.e., to set $\phi_0 = 0$),
even though $a(1) \approx 1$ suggests that one would not need $E_0$.

\subsection{Complementing polynomials}

Quantum signal processing does not use $F(t)$ itself, but rather a certain matrix element of it.
Hence, it is important to know what a matrix element can be.
Let us first introduce classes of Laurent polynomials.
\begin{definition}
    A (Laurent) polynomial with real $\RR$ coefficients is called a {\bf real} (Laurent) polynomial.
    The {\bf degree} of a Laurent polynomial is the maximum absolute value of the exponent of the variable whose coefficient is nonzero.
    A Laurent polynomial $f(z)$ is {\bf reciprocal} if $f(z) = f(1/z)$, or {\bf anti-reciprocal} if $f(z) = - f(1/z)$.
    A Laurent polynomial function $f: \CC \setminus \{0\} \to \CC$ is {\bf real-on-circle} if $f(z) \in \RR$ for all $z \in \CC$ of unit modulus.
    A real-on-circle Laurent polynomial $f(z)$ is {\bf pure} if $f(z)$ is real reciprocal 
    or $i f(z)$ is real anti-reciprocal.
\end{definition}
The term ``pure'' is because \emph{ 
a Laurent polynomial $f(z)$ with complex coefficients is  real-on-circle
if and only if both
\begin{align}
f_+(z) := \frac{f(z) + f(1/z)}{2} \text{ and } f_-(z) := \frac{f(z)-f(1/z)}{2i}
\end{align}
are real Laurent polynomials.
}
(\emph{Proof}: Write $f(e^{i\theta}) = \sum_j a_j e^{i j \theta}$, with the complex conjugate being $\sum_j \overline{a_j} e^{-ij\theta} = \sum_j \overline{a_{-j}} e^{ij\theta}$. Thus, $a_j = \overline{a_{-j}}$, and the claim follows.)
This simply rephrases the fact that a real-valued function $\theta \mapsto f(e^{i\theta})$ 
has a trigonometric function series with real coefficients.
Hence, for any real-on-circle Laurent polynomial $f(z)$, 
it is real if and only if it is reciprocal.
In addition, a real and reciprocal Laurent polynomial is real-on-circle.
That is, among the three properties, real, real-on-circle, and reciprocal, 
any two imply the third.
Note that a real-on-circle Laurent polynomial is not necessarily real,
and a real Laurent polynomial is not necessarily real-on-circle.
Also, note that real-on-circle Laurent polynomials form an algebra over the real numbers.

We are now ready to state a sufficient condition under which a complex polynomial qualifies 
to be a matrix element of some $F(t) \in \LP_n$.
We think of $a(z)$ and $b(z)$ below as the real and imaginary parts of a complex function, respectively.
A reader might want to compare the following lemma with the first paragraph of \cref{sec:parity}.
\begin{lemma}
Let $a(z)$ and $b(z)$ be real-on-circle Laurent polynomials of degree at most $n$
such that $a(\eta)^2 + b(\eta)^2 < 1$ for all $\eta \in U(1)$.
%Then, there exist pure real-on-circle Laurent polynomials $c(t)$ and $d(t)$ of degree at most $2n$ in $t$
%such that $a(t^2)^2 + b(t^2)^2 + c(t)^2 + d(t)^2 = 1$.
If $a(z)^2 + b(z)^2$ is reciprocal (e.g., $a(z)$ and $b(z)$ are pure),
then there exist pure real-on-circle Laurent polynomials 
$c(z) = c_+(z)$ and $d(z)=id_-(z)$ of degree at most $n$
such that $a(z)^2 + b(z)^2 + c(z)^2 + d(z)^2 = 1$.
\label{lem:complementing}
\end{lemma}

The conditions in the lemma on reciprocity are
due to a technical reason in the proof.
If there were some other reason under which one can guarantee the existence of the complementing polynomials,
it would be possible to use them in the algorithm below,
and the scope of the input functions in our algorithm would be enlarged;
the existence of the complementing polynomials is more important than the reciprocity constraints.
However, we note that the reciprocity conditions are 
not severe restrictions since any periodic function is the sum of
an even and an odd function, and one can ``combine'' two functions 
by ``flexible'' quantum signal processing~\cite{Low2017flexible}.

\begin{proof}
The Laurent polynomial $1-a(z)^2-b(z)^2$ is reciprocal real of degree $n'$ that is at most $2n$;
the leading terms of $a(z)^2$ and $b(z)^2$ might cancel each other so that $n' < 2n$.
Due to the reciprocity, there are $2n'$ roots in total with multiplicity taken into account
and any root $r$ must come in a pair $(z,z^{-1})$
where one is inside the unit disk and the other outside the unit disk,
but neither is on the unit circle.
We collect all the roots inside the unit disk:
\begin{align}
\mathcal D = \left[ r \in \CC ~:~ 1 - a(r)^2 - b(r)^2 = 0, ~ |r| < 1 \right]. \label{eq:innerRoots}
\end{align}
This is a list rather than a set as we take the multiplicities into account;
$\mathcal D$ has exactly $n'$ elements.
Consider a factor of $1-a(z)^2-b(z)^2$:
\begin{align}
    e(z) &= z^{-\lfloor n'/2 \rfloor} \prod_{r \in \mathcal D} (z-r).\label{eq:ez}
\end{align}
The monomial in front of the product is to balance the greatest exponent of $z$ with the least exponent;
the degree of $e(z)$ is $\lceil n'/2 \rceil$.
The list $\mathcal D$ is closed under complex conjugation due to the reality of $1-a(z)^2-b(z)^2$,
and hence $e(z)$ is a real Laurent polynomial.

Then, the product $e(z) e(1/z)$ is real reciprocal and has degree $n'$.%
\footnote{
The degree of a Laurent polynomial in our definition
is only subadditive under multiplication of two Laurent polynomials.
For example, the product $(z-1)(z^{-1} -2)$ has degree one.
}
Now the two Laurent polynomials $e(z)e(1/z)$ and $1-a(z)^2-b(z)^2$ have the same roots.
Therefore they differ by a factor of $cz^k$ for some nonzero number $c$ and an integer $k$,
but the reciprocity fixes $k =0$ and the reality puts $c$ into $\RR$.
That is, 
\[
\alpha = \frac{1-a(z)^2-b(z)^2}{e(z)e(1/z)} \in \RR .
\]
Evaluating this expression at $z=1$, we see that $\alpha$ is positive.
Thus, we finish the proof by observing
\begin{align}
1-a(z)^2-b(z)^2 = \alpha e(z)e(1/z) 
= 
\left(\frac{e(z) + e(1/z)}{2} \sqrt \alpha \right)^2 + \left(\frac{e(z) - e(1/z)}{2i} \sqrt \alpha \right)^2 .
\label{eq:czdz}
\end{align}
Both the reciprocal ($c(z)$) and anti-reciprocal ($d(z)$) combinations 
have degree $\le \lceil n'/2 \rceil \le n$.
\end{proof}
Note that given $a(z)$ and $b(z)$ 
the complementing Laurent polynomials $c(z),d(z)$ are not unique in general.
As long as the joint of a conjugate-closed list $\mathcal D$ 
and its reciprocal $[ 1/r \in \CC ~:~r \in \mathcal D ]$
is the list of all the roots of $1-a(z)^2-b(z)^2$,
we can construct $c(z)$ and $d(z)$ satisfying the conditions in the lemma.

\section{Efficient implementation with bounded precision}

In this section we consider an algorithm to find interspersing single-qubit unitaries
given a complex function $A(e^{i\varphi}) + i B(e^{i\varphi})$.
The algorithm consists of two main parts:
first, we have to find an $SU(2)$-valued function of $\varphi$ 
such that a particular matrix element is the input function. It suffices to find a good approximation.
Second, we have to decompose the $SU(2)$-valued function into a product of primitive matrices.
We have already given constructive proofs for both the steps,
but we tailor the construction so that numerical error is reduced and traceable.
We will outline our algorithm first, deferring certain details to the next subsection.
The computational complexity will be analyzed subsequently.

\begin{enumerate}
\item[Input:]
A real parameter $\epsilon \in (0,\frac{1}{100})$,
a list of $2N+1$ complex numbers 
$\zeta_k$ ($k = -N, \ldots, N$)
specified using at most $\log_2(100 N / \epsilon)$ bits 
in the floating point representation,
and two bits $p_\text{re}$ and $p_\text{im}$.

Here, the list is the Fourier coefficients $\zeta_k$ for frequencies between $-N$ and $N$ 
of a complex-valued $2\pi$-periodic function 
$A(e^{i\varphi}) + i B(e^{i\varphi}) = \sum_{k = -N}^N \zeta_k e^{i k \varphi}$
subject to conditions that
(i) each of real-valued functions $A(e^{i\varphi})$ and $B(e^{i\varphi})$ 
has definite parity (even or odd parity as functions of $\varphi$),
recorded in the two bits $p_\text{re}, p_\text{im}$,
and (ii) $ A(e^{i\varphi})^2 + B(e^{i\varphi})^2 \le 1$ for any real $\varphi$.

\emph{
The function $\varphi \mapsto A(e^{i\varphi}) + i B(e^{i\varphi})$
must be sufficiently close to an ultimate target function,
where the latter is, strictly speaking, not a part of the input for us.
This approximation has nothing to do with the algorithm below,
but should be analyzed independently for each quantum signal processing problem.
}

\item[Output:]
A unitary $E_0 \in SU(2)$ and an ordered list of $2\times 2$ hermitian matrices
$P_1,\ldots,P_{2n}$ where $n \le N$.

Here, each $P_m$ is a (approximate) rank-one projector 
represented by $\Omega(\log(N/\epsilon))$ bits of precision
and defines the primitive matrix  
$E_m(t) = t P_m + t^{-1} (I - P_m)$.
When $t = e^{i \phi} \in U(1)$,
it holds that
$\left| A(t^2) + i B(t^2) - \bra + E_0 E_1(t) \cdots E_{2n}(t) \ket + \right| \le 30 \epsilon$.

\item[Time:]
The computational time complexity is $\calO(N^3 \polylog(N/\epsilon))$
on a random-access memory machine.

\end{enumerate}

Alternatively, an input may be a list of function values from which Fourier coefficients can be computed.
Having Fourier coefficients for frequencies between $-N$ and $N$
is equivalent to having a Laurent polynomial $A(z)+ iB(z)$ of degree at most $N$.

Our \cref{lem:complementing} would allow more general input functions 
where the real and imaginary parts are not necessarily of definite parity,
but we restrict our algorithm to inputs of definite parity,
since we require $A^2 + B^2$ to be of definite parity
which may not be satisfied after the rational approximation in Step~1 below 
if individual components are not of definite parity.
As mentioned before, this parity constraint is not too restrictive 
since quantum signal processing can be easily adopted to handle functions of indefinite parity~\cite{Low2017flexible}.

\subsection{Algorithm}

\begin{enumerate}

\item[1.] Compute rational real-on-circle Laurent polynomials $a(z), b(z)$ from 
$(1-10\epsilon)A(z), (1-10\epsilon)B(z)$ 
    by taking rational number approximation for each coefficient up to an additive error of $\epsilon / N$.
    If a coefficient is smaller than $\epsilon/N$ in magnitude, then it must be replaced by zero.
    The polynomials $a(z)$ and $b(z)$ are stored as lists of rational numbers (not floating point numbers).
    This step in general decreases the Laurent polynomial degree from $N$ to $n$
    since some small coefficients may be approximated by zero.
    
\item[2.] To additive accuracy $2^{-R}$ where $R \gtrsim 2 N \log_2 (N/\epsilon)$, 
	find all roots of $1- a(z)^2 - b(z)^2$.
	See \cref{eq:Rest} for a rigorous bound on $R$.
    The roots are stored as floating point numbers.
    From now on all real arithmetic will be performed using $R$-bit floating point numbers.

\item[3.] Evaluate the complementary polynomials computed from the roots of Step~2 
	according to \cref{lem:complementing} at points of $T$ where
    \begin{align}
    T := \{ e^{2\pi i k / D} ~|~ k = 1,\ldots, D \}
    \end{align}
    and $D$ is a power of 2 that is larger than $2n+1$.
    One should \emph{not} expand $c(z),d(z)$ before evaluation,
    but should substitute numerical values for $z$ with accuracy $2^{-R}$
    in the factorized form of $e(z)$ in \cref{eq:ez},
    and then read off the real ($c(z)$) and imaginary ($d(z)$) parts.
    
\item[4.] Set $F(z) = a(z)I + b(z)iX + c(z)iY + d(z)iZ$.
	Compute the discrete fast Fourier transform of the function value list to obtain
 	\begin{align}
  		C^{(2n)}_{2j} = \int_0^{2\pi} \frac{\rd \theta}{2\pi} ~ e^{-ij\theta} F(e^{i\theta})
	\end{align}
 	for $j = -n, -n+1,\ldots, n-1,n$.
 	(In exact arithmetic, we would have $F(z) = \sum_{j=-n}^n C^{(2n)}_{2j} z^j$.)
\item[5.] Set $F^{(2n)}(t) = \sum_{j=-n}^n C^{(2n)}_{2j} t^{2j}$.
For $m = 2n, 2n-1,\ldots,2, 1$ sequentially in decreasing order,
	(i) compute a primitive matrix $E_m(t)$ by
 	\begin{align}
     &E_m(t) = t P_m + t^{-1} (I-P_m) = t(I-Q_m) + t^{-1} Q_m\\
     &\text{ where }\quad P_m = \frac{C^{(m)\dagger}_{m} C^{(m)}_{m}}{\Tr( C^{(m)\dagger}_{m} C^{(m)}_{m})},\quad
     Q_m =\frac{C^{(m)\dagger}_{-m} C^{(m)}_{-m}}{\Tr( C^{(m)\dagger}_{-m} C^{(m)}_{-m})}, \nonumber
 	\end{align}
 	and (ii) compute the coefficient list of $F^{(m-1)}(t) = F^{(m)}(t) E^\dagger_m(1/t)$ by
 	\begin{align}
 	C^{(m-1)}_{k} = C^{(m)}_{k-1} Q_m + C^{(m)}_{k+1} P_m
 	\end{align}
 	where $k = -m+1, -m+3,\ldots,m-3,m-1$.
\item[6.] 
	Output $E_0 = C^{(0)}_0$ and $P_1, \ldots, P_{2n}$ using $\log_2(20N/\epsilon)$-bit floating numbers. Then,
	\begin{align}
		\bra + E_0 E_1(t) \cdots E_{2n}(t) \ket +
	\end{align} 
	is 
    $30\epsilon$-close to $A(t^2) + iB(t^2)$ for all $t \in U(1)$.
%    (If $R$ is sufficiently large that the numerical error in total is smaller than $\epsilon$,
%    then we can guarantee $15\epsilon$-closeness of the output to the input.)
\end{enumerate}

\subsection{Further details and analysis}

\paragraph*{Step 1.}

Let the rational approximation be performed 
by truncating the binary expressions of real numbers.
To inherit parities, the rational approximation should be done only for terms with nonnegative powers of $z$,
from which we should infer the negative power terms.
Then, $a(z)$ and $b(z)$ satisfy 
\begin{enumerate}
\item[(i)] $a(z)^2 + b(z)^2$ is a real reciprocal Laurent polynomial,
\item[(ii)] $|a(z) + ib(z) - A(z) - i B(z) | \le 26 \epsilon$ for $z \in U(1)$,
\item[(iii)] $a(z)^2 + b(z)^2 \le 1 - \epsilon$ for $z \in U(1)$, and
\item[(iv)] every nonzero coefficient in $a(z)$ or $b(z)$ has magnitude $\ge \epsilon / N$.
\end{enumerate}
The first and the fourth conditions are clear by construction.
The second condition is because 
$|A(z) - a(z)| \le |A(z) - (1-10\epsilon)A(z)| + |(1-10\epsilon)A(z) - a(z)| \le 13\epsilon$
and similarly $|B(z) - b(z)| \le 13 \epsilon$.
The third is because $|a(z)+ib(z)| \le |(1-10\epsilon)(A(z)+iB(z)) - a(z)+ib(z)| + (1-10\epsilon) \le 1-\epsilon$.

The reason we speak of rational Laurent polynomials is mainly for the convenience of analysis,
as its evaluation can be made arbitrarily accurate
since the coefficients of $a(z)$ and $b(z)$ are exact;
each coefficient is stored as a rational number, a pair of integers, rather than a floating point number.
In a concrete implementation of our algorithm,
floating point numbers that are carefully handled may substitute rational numbers.
If $\mu = f \times 2^d$ is a real number where $\frac 1 2 \le f < 1$ and $d \in \mathbb Z$,
then the rational approximation of $\mu$ may be $\tilde \mu = 0.b_1 b_2 \cdots b_p \times 2^d$ for some $p$
where $b_j$ are bits in the binary representation of $f$.
Some care may be needed in order not to discard any bits of $\tilde \mu$ in arithmetic.
For example, when $\tilde \mu$ is added to another floating number of $p' > p$ bits of precision,
then $\tilde \mu$ should be mapped to a floating number of whatever needed bits of precision 
by padding zeros.
In practice, this should cause hardly any complication 
since most high precision arithmetic libraries
(e.g. The GNU Multiple Precision Arithmetic Library)
treat inputs as exact numbers,
but control the precision of the result according to other rules set by a user.
The number of bits to represent all the rational coefficients is $\calO( N \log(N/\epsilon))$.

\paragraph*{Step 2.}
There exists a root-finding algorithm with computational complexity $\tilde \calO(n^3 + n^2 R)$
under the assumption that all the roots have modulus at most 1~\cite{Pan1996}.
In our case, the rational Laurent polynomial $p(z) =1- a(z)^2 - b(z)^2$ does not satisfy the modulus condition;
however, this is a minor problem.
Every coefficient of $p(z)$ is the Fourier coefficient of the periodic function $p(e^{i\theta}) < 1$,
and hence is bounded by $1$.
By the condition~(iv) of Step~1, the leading coefficient of $p(z)$ is $\Omega((\epsilon/N)^2)$ in magnitude.
(The reason is as follows.
Since our rational approximation is by truncating binary expressions of real numbers,
the denominator of any coefficient is a power of 2 and is at most $2^{\lceil \log_2 (N / \epsilon) \rceil}$.
Hence, $4^{\lceil \log_2 (N / \epsilon) \rceil} p(z)$ has integer coefficients.)
Say the polynomial $p(z)= 1-a(z)^2 - b(z)^2$ has degree $n'$,
which may be less than $2n$ even if $a(z)$ and $b(z)$ have degree $n$.
Converting $p(z)$ into a monic polynomial $q(z)$ 
(after multiplying by $z^{n'}$ and a normalization factor),
we have $q(z) = z^{2n'} + \sum_{j=0}^{2n'-1} q_j z^j$ with $|q_j| \le \calO( (N/\epsilon)^2 )$.
Note that $q(z)$ has (exactly represented) rational coefficients.
If $q(z_0)=0$ with $|z_0| > 1$, then
\begin{align}
|z_0|^{2n'} \le \sum_{j=0}^{2n'-1} |q_j| |z_0|^j \le \calO(N^2 n'|z_0|^{2n'-1} / \epsilon^2)
\end{align}
implying $|z_0| \le \calO(N^3/\epsilon^2)$. (If $|z_0| \le 1$, this is trivial.)
We can use the algorithm of Ref.~\cite{Pan1996} after rescaling $z$ by a known factor.
The overhead due to the potential loss of precision from the rescaling is negligible 
since $R = \Omega(N \log (N/\epsilon))$.

\paragraph*{Step 3.}
We need to evaluate $e(z)$ of \cref{eq:ez} that is defined by the roots found in Step~2.
For $z \in U(1)$, the following \cref{lem:root-separation} guarantees that any root of 
$1-a(z)^2 - b(z)^2$ is
at least $\epsilon/(4 N^2)$-away 
(or $\Omega(1/N^2)$-away if $|1-A(z)^2-B(z)^2| = \calO(\epsilon)$)
from the unit circle.
In particular, with the prescribed accuracy $2^{-R}$
there is no numerical ambiguity to determine whether a root is inside the unit disk.
That is, the list $\mathcal D$ of \cref{eq:innerRoots} 
can be obviously computed under our bounded precision arithmetic.

Let us analyze the evaluation accuracy of $e(z)$ for $z \in U(1)$ more closely.
When evaluating a linear factor $z - r$ of $e(z)$
where both $z$ and $r$ are accurate up to additive error $2^{-R}$,
the number of lost significant bits is $\calO(\log(N/\epsilon))$
which is negligible compared to $R$.
The function value of $e(z)$ is thus evaluated accurately up to relative error 
$2^{-R + \calO(\log(N/\epsilon))}$.
Hence, the function value of $c(z)$ (the real part of $e(z)\sqrt{\alpha}$) 
and $d(z)$ (the imaginary part of $e(z)\sqrt{\alpha}$) by \cref{eq:czdz}
are determined up to additive error $2^{-R + \calO(\log(N/\epsilon))}$.

More concretely but still loosely, let us assume $24 N^2 2^{-R}\epsilon^{-1} \le 1/(16N)$.
The relative error of a linear factor $z-r$ is at most $2 \cdot 2^{-R} (4 N^2 / \epsilon)$.
The factor of $e(z)$ for the complex roots can be evaluated 
to relative error at most $3 \cdot 2 \cdot 2^{-R} (4 N^2 / \epsilon)$.
There are at most $N$ factors in $e(z)$,
so the value of $e(z)$ is determined up to relative error
$\delta = (1+ 24 N^2 2^{-R}\epsilon^{-1})^N - 1 \le 48 N^3 2^{-R}\epsilon^{-1} \le 1/8$.
This in turn gives an upper bound 
$200 N^3 2^{-R}\epsilon^{-1}$
on the additive error of the real part $c(z)$ and the imaginary part $d(z)$
since they have magnitude at most 1 on the unit circle.

\begin{lemma}\label{lem:root-separation}
If a real-on-circle Laurent polynomial $f(z)$ of degree $d \ge 1$ satisfies 
$0 < m \le f(z) \le M$ for all $z \in U(1)$,
then every zero of $f$ is at least $m/(4 M d^2)$-away from $U(1)$.
\end{lemma}
\begin{proof}
Pick any root $z_0$ and choose the closest point $u \in U(1)$ so that $f(z_0 = u + \eta) = 0$.
We will lower bound the magnitude of $\eta$.
If $|\eta| \ge 1/(2d)$, then there is nothing to prove,
so we assume $|\eta| < 1/(2d)$.
Since $f$ is analytic except at $z=0$,
the Taylor series of $f$ at $u$ converges at $z_0 = u + \eta$.
\begin{align}
0 = f(u+\eta) = \sum_{k \ge 0} \frac{f^{(k)}(u)}{k!} \eta^k .\label{eq:fueta}
\end{align}
Let us estimate the magnitude of derivatives at $u \in U(1)$.
Since the coefficients $a_j$ of the polynomial
$f(u) = \sum_{j=-d}^d a_j u^j$ are the Fourier coefficients,
meaning that $a_j$ is a ``weighted'' average of the function values on the unit circle,
we know $|a_j| \le M$.
Thus
\begin{align}
|f^{(k)}(u)| \le 2d \cdot M \cdot d(d+1) \cdots (d+k-1);
\end{align}
there are $2d$ terms (or one more if $k=0$ in which case the inequality is true anyway)
and the maximum absolute value of the exponent increases by at most~1
every time we differentiate due to the negative degree term.
Therefore,
\begin{align}
m 
\le |f(u)| 
\le \sum_{k \ge 1} 2d M \binom{d+k-1}{k} |\eta|^k 
= 2dM \left( \frac{1}{(1-|\eta|)^d} - 1\right) 
\le 2 d M \cdot 2d |\eta|
\end{align}
where 
the first inequality is the assumption, 
the second inequality is by rearranging \cref{eq:fueta} and applying triangle inequality,
and in the last inequality we use 
the fact that $(1-x)^{-d} - 1$ is a convex increasing function of positive $x$ 
and valued at most 1 at $x = 1/(2d)$ for $d \ge 1$.
\end{proof}

\paragraph*{Step 4.}
This is essentially expanding the polynomials $c(z)$ and $d(z)$ found in the root finding step,
but we use the fast Fourier transform (FFT) for its better accuracy.
It has been shown~\cite{Ramos1971} that the FFT on a $k$-component input $F$ where each component 
(that is assumed to be a complex number in \cite{Ramos1971})
is accurate to relative error $\delta$, gives Fourier coefficients $\tilde{\hat F}_\omega$ with error
\begin{align}
\max_\omega | \tilde{ \hat{ F_\omega} } - \hat F_\omega | \le \calO( k^{-1/2} \log k ) \cdot \delta \sqrt{\frac{1}{k} \sum_\omega |\hat F_\omega|^2 }
\label{eq:FFT-bound}
\end{align}
where $\hat F_\omega = k^{-1} \sum_{\ell=1}^k e^{i\ell \omega/k} F(e^{i\ell/k})$
is the true Fourier spectrum. (Note the normalization factor $k^{-1}$ here.)
In our case, $F$ consists of $2$-by-$2$ matrices,
and \cref{eq:FFT-bound} also holds with Frobenius or operator norms in place of absolute values.
Since the input ``vector''~$F$ in this Step is a list of unitary matrices,
the root-mean-square factor is $\calO(1)$,
and the distinction between relative and absolute error is immaterial.
By the analysis of Step~3 above, $\delta$ is $2^{- R + \calO(\log(N/\epsilon))}$.
Thus, the (additive) error $\delta_{2n}$ in any Fourier coefficient $C^{(2n)}_{2j}$ 
is at most $2^{- R + \calO(\log(N/\epsilon))}$.
The error here is the operator norm of the difference 
between the computed and the true $C^{(2n)}_{2j}$.

Crude and concrete bounds can be obtained more directly 
(without using \cref{eq:FFT-bound}):
The coefficients of $a(z)$ and $b(z)$ are exactly known.
Those of $c(z)$ and $d(z)$ are computed from the value table of $e(z)$ from the previous Step,
which has entry-wise additive error $\delta \le 48 N^3 2^{-R} \epsilon^{-1}$.
The (slow) discrete Fourier transform on a $k$-component vector $v$
is the matrix multiplication by $V = k^{-1/2} U$ for a $k \times k$ unitary $U$.
(Hence, $\| V \| \le k^{-1/2}$.)
If we compute the input-independent trigonometric factors in the FFT,
often called the twiddle factors,
accurately up to additive error $2^{-R}$,
then $k^{-1/2}U$ is accurate up to additive error $\delta_{FFT} \le k^{1/2} 2^{-R}$ in operator norm.
(This bound is loose compared to that in the previous paragraph,
since in this paragraph we do not consider the fast Fourier transform 
that is numerically more stable.)
Since $k \le 2N+1$, the conversion from $\ell_\infty$-norm to $2$-norm incurs a factor of at most $\sqrt{2N+1}$.
Hence, the additive error $\delta_{2n}$ in $C^{(2n)}_{2j}$ for any $j$ is at most
$
\| \tilde V \tilde v - V v \|_\text{max}
\le
\| \tilde V \tilde v - V v \|_2 
\le 
\| \tilde V - V \| \cdot \| \tilde v \|_2 + \| V \| \cdot \| \tilde v - v\|_2 
\le
\delta_{FFT} \sqrt{2N+1} + \delta$.
That is, 
\begin{align}
\delta_{2n} \le 400 N^{3} \epsilon^{-1} 2^{-R}. \label{eq:delta2n}
\end{align}

\paragraph*{Step~5.}
This is an implementation of \cref{thm:composition}.
Thanks to the condition~(iv) of Step~1, we know
\begin{align}
\norm{ C^{(2n)}_{\pm 2n} } \ge \epsilon / N.
\end{align}

Put $F(t^2) = E_0 E_{1}(t) \cdots E_{2n}(t)$.
Then, for any $m = 1,2,\ldots,2n$, we observe that $C^{(m)}_{\pm m}$ is the product
\begin{align}
    C^{(m)}_m = E_0 P_1 P_2 \cdots P_m, \qquad C^{(m)}_{-m} = E_0 Q_1 Q_2 \cdots Q_m .
\end{align}
In particular, by the submultiplicative rule of operator norm we have
\begin{align}
   \norm{ C^{(m)}_{\pm m} } \ge \norm{ C^{(2n)}_{\pm 2n} } \ge \epsilon / N;
\end{align}
that is, the leading coefficients never become smaller in norm through the loop over $m$.

Let $\delta_{m}$ be the maximum additive error of $C^{(m)}_k$ for any $k$.
We assume that $\delta_{m} \le \epsilon / (2N)$.
For brevity, let $C = C^{(m)}_{\pm m}$,
and let $\tilde C$
be the approximate $C$ due to numerical error.
Then, $\delta_m \le \epsilon / (2N) \le  \norm{C} / 2$
and $\norm{C}/2 \le \norm{C} - \delta_m \le \norm{\tilde C} \le \norm{C} + \delta_m \le 2\norm{C}$. Now,
\begin{align}
\norm{\tilde C^\dagger \tilde C - C^\dagger C}
&\le \norm{ \tilde C^\dagger } \norm{\tilde C - C} + \norm{\tilde C^\dagger - C^\dagger} \norm{C} \nonumber \\
&\le 2\norm{C} \cdot \delta_m + \delta_m \cdot \norm{C} \\
&= 3 \norm{C} \delta_m\nonumber\\
%%%%
\abs{\Tr(\tilde C^\dagger \tilde C) - \Tr(C^\dagger C)} 
&\le \sum_{j=0}^1 \abs{\bra j \tilde C^\dagger \tilde C - C^\dagger C \ket j}
\le 6 \norm{C} \delta_m \\
\Tr(C ^\dagger C) 
&\ge \norm{C}^2 \\
\Tr(\tilde C^\dagger \tilde C) 
&\ge \norm{\tilde C}^2 \ge \frac{1}{4} \norm{C}^2
\end{align}
Hence, we see that
the additive error $\beta_m$ in $P^{(m)}$ (or $Q^{(m)}$) is
\begin{align}
\norm{\tilde P^{(m)} - P^{(m)}}
&\le
 \frac{\norm{ \tilde C^\dagger \tilde C - C^\dagger C }}{\Tr(\tilde C^\dagger \tilde C)} 
   +
 \frac{\norm{ C^\dagger C }\abs{\Tr(C^\dagger C)-\Tr(\tilde C^\dagger \tilde C)}}{\Tr(\tilde C^\dagger \tilde C)\Tr(C^\dagger C)} 
  \nonumber\\
&\le
 \frac{4}{\norm{C}^2} \cdot 3 \norm{C}\delta_m + \frac{4}{\norm{C}^2}\cdot \frac{1}{\norm{C}^2} \cdot \norm{C}^2 \cdot 6 \norm{C}\delta_m \nonumber\\
&\le
36 N \epsilon^{-1} \delta_m.
\end{align}
In turn, assuming $\beta_m \le 1$ we have 
\begin{align}
\delta_{m-1} 
&=  
\max_k
\norm{\tilde C^{(m)}_{k-1} \tilde Q_m + \tilde C^{(m)}_{k+1} \tilde P_m
-
C^{(m)}_{k-1} Q_m - C^{(m)}_{k+1} P_m}\nonumber\\
&\le 
\max_k 
\norm{\tilde C^{(m)}_{k-1} \tilde Q_m - C^{(m)}_{k-1} \tilde Q_m}
+ \norm{C^{(m)}_{k-1} \tilde Q_m - C^{(m)}_{k-1} Q_m}
+ \text{(similar terms)} \nonumber\\
&\le
2(2 \delta_m + \beta_m) \nonumber\\
&\le
76 N \epsilon^{-1} \delta_m. \label{eq:deltaMminus1}
\end{align}
Combining \cref{eq:delta2n,eq:deltaMminus1} we conclude that
\begin{align}
\beta_0 := \delta_0 \le  400 N^{3} \epsilon^{-1} (76 N \epsilon^{-1})^{2N} 2^{-R} =: \Gamma.
\end{align}

\paragraph*{Step~6.} 
Now we have approximate $E_0$ up to additive error $\beta_0$ 
and $P_m$ up to $\beta_m$ ($m=1,\ldots,2n$).
The evaluation of $E_m(t)$ is then accurate up to $2\beta_m$ for $t \in U(1)$,
and that of $a(t^2) + i b(t^2) = \bra + E_0 E_1(t) \cdots E_{2n}(t) \ket +$
is accurate up to (suppressing $t$ for brevity) 
\begin{align}
\norm{\tilde E_0 \cdots \tilde E_{2n} - E_0 \cdots E_{2n}}
&\le
\sum_{m=0}^{2n} \norm{ E_0 \cdots E_{m-1} \tilde E_{m} \cdots \tilde E_{2n} - 
E_0 \cdots E_{m} \tilde E_{m+1} \cdots \tilde E_{2n}} \nonumber \\
&\leq
\sum_{m=0}^{2n}\left( 2\beta_m \prod_{\ell=m+1}^{2n}(1+2\beta_\ell) \right).
\end{align}
We want this to be smaller than $\epsilon$.
A sufficient condition is thus
\begin{align}
(2N+1)2\Gamma\left(1 + 2\Gamma \right)^{2N+1} &\le \epsilon \label{eq:Rest}\\
\text{or }\quad R &\gtrsim 2N\log_2(N/\epsilon).
\end{align}
All the (ad hoc) assumptions in the course of error estimation above,
are satisfied by this choice of $R$.

The correctness of the algorithm is clear from the construction and error analysis above.
The final error guarantee is from the condition~(ii) of Step~1 and \cref{eq:Rest}.

\subsection{Computational complexity}

All arithmetic in the algorithm above operates with at most $R$-bit numbers, where $R$ is chosen to be
$\Theta(N \log (N/\epsilon))$.
The number of elementary bit operations $(AND,OR,NOT)$
to perform one basic arithmetic operation ($+,-,\times,/$) on $u$-bit numbers
is upper bounded by $\calO(u~ \polylog(u))$~\cite{HarveyHoeven2018}.%
\footnote{This reference is concerned with multiplications,
but the division has essentially the same cost using the Newton's method~\cite[4.3.3~R]{Knuth}.}
Let us count the number of arithmetic operations.

\emph{Step~1}: Assuming that the coefficients of the true function $A(z)$ and $B(z)$ are given to accuracy $\calO(\epsilon/N)$,
it takes $\calO(N \log(N/\epsilon))$ arithmetic operations to find rational approximations.
There could be $\calO(\polylog(N/\epsilon))$ additional cost in full complexity in simplifying
rational numbers by Euclid's algorithm.

\emph{Step~2}: The root finding takes time $\calO(N^3+N^2 R ) \le \calO( N^3 \log(N/\epsilon))$~\cite{Pan1996};
this count includes all the cost of bit operations for high precision arithmetic.

\emph{Step~3}:
Selecting roots inside the unit disk requires $\calO(N)$ absolute value evaluations and $\calO(N)$ comparisons.
The polynomial function evaluation involves $\calO(N)$ arithmetic operations.
We need $\calO(N)$ values, resulting in $\calO(N^2)$ arithmetic operations.

\emph{Step~4}: The FFT requires $\calO(N \log N)$ arithmetic operations 
given $\calO(\log N)$ trigonometric function values,
which can be computed by Taylor expansions of order $R$, invoking $\calO(R \log N)$ arithmetic operations.
These result in arithmetic complexity $\calO( N \log(N) \log (N/\epsilon))$ for the FFT.

\emph{Step~5}: Updating the Fourier coefficient list $C^{(m)}_k$ involves $\calO(N)$ arithmetic operations,
which we do for $\calO(N)$ times, resulting in $\calO(N^2)$ arithmetic operations.

Overall, the computational complexity is $\calO( N^3 ~\polylog(N/\epsilon))$,
under the random-access memory model of computation.

Practically, it may be useful to 
run the algorithm with arithmetic precision with, say, $R_0 = 64$ bits initially,
test the final decomposition on $\calO(N)$-th roots of unity
(which takes only $\calO(N^2)$ operations with $\calO(\log(N/\epsilon))$-bit arithmetic),
and repeat the whole process under an exponential scheduling 
on the number $R_{r+1} = 2 R_r$ of bits of precision in the arithmetic,
until the test reveals that the answer is acceptable.
In this way the arithmetic uses no more than twice the number of bits of precision
that is actually needed to handle the numerical instability of our algorithm for a given input,
without knowing a tailored number beforehand.
For the upper bound $R = \calO( N \log (N/\epsilon))$ in the worst case,
there will be at most $r_\text{max} = \calO(\log N + \log \log (N/\epsilon))$ rounds,
but the overall time complexity is a constant multiple of the last round's
due to the exponential scheduling.

\section{Application to Hamiltonian simulation}
\label{sec:JacobiAnger}

In this section, we review and redo existing analysis~\cite{BCK2015,QSP}, 
complemented with a minor modification to our algorithm exploiting 
a certain structure of the eigenvalue transformation function.

Suppose we are given with a unitary $W$ whose eigenvalues $\mp e^{\pm i \theta_\lambda}$
are associated with those $\lambda$ of a Hermitian matrix $H = \sum \lambda \ket \lambda \bra \lambda $
with $\norm{H} \le 1$ as
\begin{align}
\sin \theta_\lambda = \lambda . \label{eq:sinthetalambda}
\end{align}
The correspondence between $W$ and $H$ might seem contrived at this stage,
but when $H$ is represented as a linear combination of unitaries~\cite{ChildsWiebe,BCCKS},
it is possible to construct such $W$ as a quantum circuit~\cite{LC16}.
The relation of \cref{eq:sinthetalambda} is in fact common 
whenever quantum walk is used~\cite{Childs2008,BerryChilds2009}.
(See \cref{app:jordan} for some detail.)
So, the desired transformation is
\begin{align}
f : \mp e^{\pm i \theta_\lambda} \mapsto e^{-i \tau \sin \theta_\lambda} \approx \bra + F(e^{i\theta_\lambda/2}) \ket +
\end{align}
where $F(t)$ should be constructed by quantum signal processing~\cite{QSP}.
This will implement $e^{-i \tau H}$.
Since the product of $n$ primitive matrices yields a Fourier component of frequency at most $n/2$,
$F(t)$ must consist of at least $2\tau$ factors.
(The factor of 2 is due to the half-angle in the argument of $F$.)
Note that the success probability of the post-selection is close to $1$
since $|f(e^{i\theta_\lambda})|=1$.

With $e^{i\varphi} = z$, we write
\begin{align}
\exp(i \tau \sin \varphi) = \exp\left( \tau \frac{z-z^{-1}}{2} \right) 
= \sum_{k \in \mathbb Z} J_k(\tau) z^k 
\label{eq:JacobiAnger}
\end{align}
where $J_k$ are the Bessel functions of the first kind;
one can take \cref{eq:JacobiAnger} as a \emph{definition} of the Bessel functions.
This is the Fourier series of $\varphi \mapsto \exp(i \tau \sin \varphi)$.
The substitution $\tau \to -\tau$ and $z \to -z$ together with the uniqueness of the Fourier series
implies $J_k(-\tau) = (-1)^k J_k(\tau)$.
Similarly, the substitution $z \to -1/z$ implies $J_{-k}(\tau) = (-1)^k J_k(\tau)$.
We separate the reciprocal and anti-reciprocal parts of the expansion as
\begin{align}
\exp\left( \tau \frac{z-z^{-1}}{2} \right) 
=
\underbrace{\sum_{k \in 2\ZZ} J_{k}(\tau) \frac{z^k + z^{-k}}{2}}_{A(z)} 
~ + ~ i 
\underbrace{\sum_{k \in 2\ZZ+1} J_{k}(\tau) \frac{z^k - z^{-k}}{2i} }_{B(z)}.
\end{align}
This expansion, called the Jacobi-Anger expansion, converges absolutely at a superexponential rate.
We can use the steepest descent method~\cite{Boyd1994}
which is generally applicable.
Expressing the Fourier transform as a contour integral we see
\[
J_k(\tau) = \frac{1}{2\pi i} \int_{C} \frac{\rd z}{z^{k+1}}  e^{(\tau/2)(z - z^{-1})}
\]
where $C$ is the unit circle.
Since the integrand is analytic except for $z=0$,
we may deform~$C$.
For $2k > \tau > 0$,
$z \approx 2k/\tau > 1$ is a saddle point for the absolute value of the integrand.
We take $C$ to be a circle through this point, to have that 
\begin{align}
|J_k(\tau)| &\le \left(\frac \tau {2k}\right)^k \int_0^{2\pi} \frac{\rd \theta}{2\pi} \exp[ (k - (\tau^2/4k))\cos\theta ]
\le \left( \frac{e\tau}{2k}\right)^k  \quad \text{ for } 2k > \tau > 0,  \nonumber \\
|J_k(\tau)| &\le \left( \frac{e|\tau|}{2|k|}\right)^{|k|} \quad \text{ for } k \in \ZZ \setminus \{0\}, \tau \in \RR. \label{eq:JAbound}
\end{align}
It is important that the convergence of the series depends
on the size of the region on which the function is analytic
--- this is a general fact~\cite{Boyd1994}.
For the Jacobi-Anger expansion the function is almost entire,
and the convergence is superexponential.
Note that \cite[9.1.62]{AS1964} asserts $|J_k(\tau)| \le |\tau/2|^{|k|} / |k| !$ 
for any $\tau \in \mathbb R$ and $k \in \mathbb Z$
which is tighter than \cref{eq:JAbound}.

Now, a partial sum of the Jacobi-Anger expansion can be written as
\begin{align}
\sum_{k : |k| \le N} J_k(\tau) z^k
&=
\underbrace{J_0(\tau) + \sum_{\text{even } k : 2 \le k \le N} J_{k}(\tau) (z^k + z^{-k}) }_{\tilde A(z)} 
~+~i 
\underbrace{\sum_{\text{odd } k : 1 \le k \le N} J_{k}(\tau) \frac{z^k - z^{-k}}{i}}_{\tilde B(z)}.
\end{align}
This is $\epsilon$-close to the full expansion if $N = \Omega( |\tau| + \log(1/\epsilon))$
by \cref{eq:JAbound} for any $z \in U(1)$.%
\footnote{The proof of this is along the same lines as proving 
the convergence of Taylor series for e.g.~the exponential function,
and is left to the reader.}
Numerical experiments suggest that the bound is quite tight and it suffices to choose
\begin{align}
N \approx \frac{e}{2} |\tau| + \ln (1/\epsilon) \approx 1.36 |\tau| + 2.30 \log_{10}(1/\epsilon).
\end{align}

The Laurent polynomials $\tilde A(z)$ and $\tilde B(z)$ are pure real-on-circle.
Applying our algorithm, we obtain real-on-circle Laurent polynomials 
$a(z) = a_+(z), b(z) = i b_-(z), c(z)= i c_-(z), d(z) = d_+(z)$.
The pure Laurent polynomials $c(z), d(z)$ are calculated by \cref{lem:complementing} 
where we choose $c(z)$ to be anti-reciprocal and $d(z)$ reciprocal.
(This choice is to have our results in the same convention as those of Ref.~\cite{LowMethod}.)
Note that every exponent of $z$ of the polynomial $1-a(z)^2-b(z)^2$ here, whose roots must be computed,
is even since $a(z)$ has only even exponents and $b(z)$ has only odd exponents,
so it is always better to feed a Laurent polynomial $g$ of degree $n$, instead of $2n$, where 
\begin{align}
g(z^2) = 1-a(z)^2 - b(z)^2
\end{align}
into the root finding routine.
Given the expanded form of $1-a(z)^2 - b(z)^2$, it takes no effort to find $g$.
In this case, the intermediate polynomial $e(z)$ in \cref{eq:ez} of \cref{lem:complementing} is
\begin{align}
e(z) = \prod_{c \in \mathbb C ~:~ g(c) = 0,\, |c| < 1 } \left( z - \frac{c}{z} \right).
\end{align}

We have implemented our algorithm with constants chosen as above
using Wolfram Mathematica~11,
and measured the running time as a function of $\tau$ for two fixed values of~$\epsilon$.
The result is shown in \cref{fig:ja}.
The running time scales asymptotically as the cubic of~$\tau$ as expected.
We used internal routines of Mathematica for the rational approximation, 
high precision arithmetic, root finding, and Fourier transform.
The computing was by Microsoft Surface Book with Intel Core i7-6600U at 2.6~GHz and 16~GB of RAM.

\begin{figure}[h]
    \centering
    \includegraphics[width=0.8\textwidth]{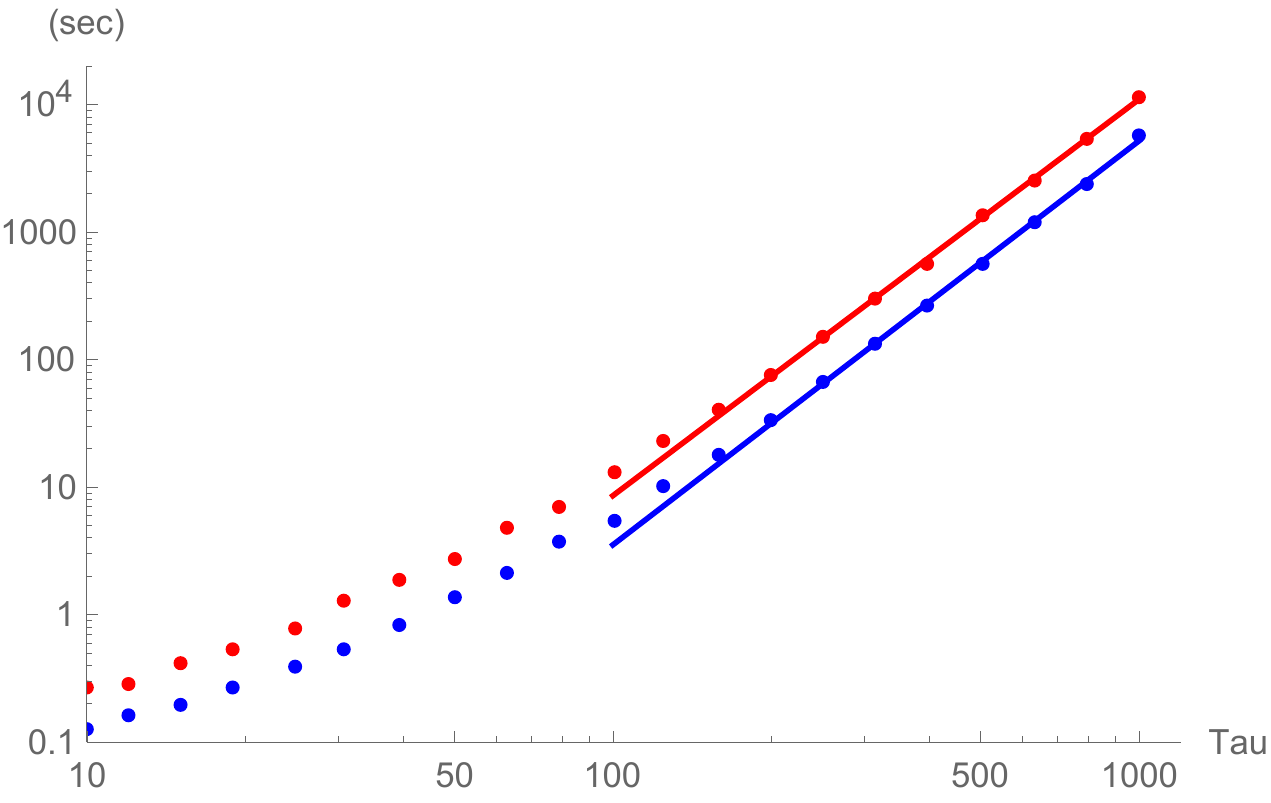}
    \caption{Running time of our algorithm for the Jacobi-Anger expansion as a function of $\tau$,
    implemented in Wolfram Mathematica~11.
    The upper red data set has $\epsilon = 10^{-9}$ and the lower blue $\epsilon = 10^{-4}$.
    The number of decimal digits in the intermediate steps
    is $(n/3) \ln(n/\epsilon)$ where $n$ is the degree of the input Laurent polynomial
    and the factor of 3 is empirically chosen
    to make the numerical error negligible.
    The straight lines represent functions $\mathrm{const}\cdot \tau^\gamma$ of exponents 
    $\gamma=3.11$ for $\epsilon = 10^{-9}$ (red) 
    and $\gamma=3.17$ for $\epsilon = 10^{-4}$ (blue).
    The top right data point has 2172 primitive matrices in the decomposition.
    }
    \label{fig:ja}
\end{figure}

\section{Application to Matrix inversion}
\label{sec:inverse}

While there are slightly more efficient implementations of matrix inversion problems~\cite{HHL} 
using quantum signal processing~\cite{GSLW2018},
here we contend ourselves with an eigenvalue transformation perspective.
The techniques of Refs.~\cite{GSLW2018} reduces the number of ancilla qubits by one or two,
and hence relieves some burden of implementing controlled unitary,
but the underlying mathematics, 
regarding polynomial approximations and finding interspersing single-qubit unitaries,
is unchanged.

Suppose a hermitian matrix $H$ of norm 1 that we wish to invert 
is block-encoded in a unitary $W$ so that $W$ has eigenvalues $\mp e^{\pm i\theta_\lambda}$
associated with an eigenvalue $\lambda$ of $H$.
This encoding is the same as in the Hamiltonian simulation above.
The condition for $H$ being hermitian is not too restrictive since,
for any matrix $M$,
an enlarged matrix
$\ket 0 \bra 1 \otimes M + \ket 1 \bra 0 \otimes M^\dagger$
is always hermitian.
Then, we want eigenvalue transformation $\mp e^{\pm i\theta_\lambda} \mapsto 1/\sin \theta_\lambda$.
As we should not invert a singular matrix, we assume that eigenvalues of $H$ 
are bounded away from zero by $1/\kappa$ where $\kappa \ge 1$ is the condition number of $H$.
(Stricly zero eigenvalues are fine if we are interested in a pseudo-inverse.)
Thanks to the condition number assumption, we need to find a polynomial approximation to the function
$\mp e^{\pm i\theta_\lambda} \mapsto 1/\sin \theta_\lambda$ 
that is good for values $\sin \theta_\lambda$ away from zero by $1/\kappa$.
For this purpose, there is a useful polynomial~\cite{ChildsKothariSomma2015Linear}:
\begin{lemma} \label{lem:approxinverse}
Let $\epsilon > 0$, $\kappa \ge 1$ and $z \in U(1)$.
Suppose integers $b \ge b' \ge 1$ satisfy $b \ge \kappa^2 \ln(2/\epsilon)$ and $b' \ge \sqrt{b \ln(8/\epsilon)}$.
For $\sin \varphi = (z-z^{-1})/(2i) \in \RR $ with $|\sin \varphi| \ge 1/\kappa > 0$,
we have
\begin{align}
\left| 
\frac{2i}{\kappa(z-z^{-1})}
-
\underbrace{\frac{2i}{2^{2b}\kappa(z-z^{-1})} \sum_{k=-b'}^{b'}\binom{2b}{b+k}(1-z^{2k})}_{f(z)}
\right|
&\le 
\epsilon. \label{eq:ckspoly}
\end{align}
Moreover, for all $z \in U(1)$ we have $|f(z)| \le 2b'/\kappa$.
\end{lemma}

The function $f(z)$ is a genuine Laurent polynomial%
\footnote{The polynomial of \cref{eq:ckspoly} is the same as that in 
\cite[Lemma~17-19]{ChildsKothariSomma2015Linear}.
The bound there is similar to ours, but the polynomial degree is worse than ours by a factor of $\log \kappa$.
This difference is due to a different normalization --- 
we approximate $1/(\kappa x)$ rather than $1/x$.
Our analysis might look simpler, but it is not due to a ``new'' approach;
the ``difference'' is only in the usage of exponential functions rather than trigonometric functions.
} 
since the sum vanishes at $z = \pm 1$.

\begin{proof}
First, let $b \ge b' \ge 1$ be any integers, and let $z = e^{i\varphi} \in U(1)$ be any complex number.
Then,
\begin{align}
\left| 
\underbrace{
1 - \left( \frac{z+z^{-1}}{2} \right)^{2b}
}_{g(z)} 
- 
\underbrace{
\frac{1}{2^{2b}}\sum_{k=-b'}^{b'}\binom{2b}{b+k}(1-z^{2k}) 
}_{h(z)}
\right| 
\le 4e^{-b'^2/b} \label{eq:binomial}
\end{align}
because
$\left| 2^{-2b}\sum_{k~:~|k| > b'} \binom{2b}{b+k}(1-z^{2k}) \right| 
\le 2^{-2b+1}\sum_{k~:~|k| > b'} \binom{2b}{b+k}$
and Hoeffding's inequality on the tail of binomial probability distributions implies \cref{eq:binomial}.
If $\sin \varphi = \frac{z-z^{-1}}{2i} $ with $|\sin \varphi| \ge 1/\kappa > 0$, then
\begin{align}
| 1 - g(z) | = \left| \left( \frac{z+z^{-1}}{2} \right)^{2b} \right| &\le e^{-b/\kappa^2},
\end{align}
since $(1-\sin^2 \varphi)^{b} \le e^{-b/\kappa^2}$ whenever $|\sin \varphi| \ge 1/\kappa$.

Thus, for large $b$ we see that $1-((z+z^{-1})/2)^{2b}$ vanishes when $\sin \varphi = 0$ (i.e., $z= \pm 1$),
but is close to $1$ for $|\sin \varphi| \ge 1/\kappa$.
By \cref{eq:binomial} this function can be replaced with a lower degree polynomial function.
Indeed, for $|\sin \varphi| \ge 1/\kappa$ we have
\begin{align}
\left| \frac{2i}{\kappa (z - z^{-1})} - \frac{2i h(z)}{\kappa (z - z^{-1})} \right|
\le |1-h(z)| \le e^{-b/\kappa^2} + 4 e^{-b'^2/b},
\end{align}
which implies \cref{eq:ckspoly}.

For the last claim, we observe a chain of (in)equalities:
For any integer $k \ge 1$ we have
\begin{align*}
(z-z^{-1})(z^{k-1}+z^{k-3}+\cdots+z^{-k+3}+z^{-k+1}) &= z^k - z^{-k},\\
| (z^k-z^{-k})/(z-z^{-1}) | &\le k &\text{ for } z \in U(1),\\
|\sin^2 k\varphi /\sin \varphi | \le |\sin k\varphi / \sin \varphi| &\le k &\text{ for }\varphi \in \RR,\\
|(2-z^{2k}-z^{-2k}) / (z-z^{-1})| &\le 2k &\text{ for }z \in U(1).
\end{align*}
The function $f(z)$ is the ``average'' over $k$ of 
\[
\frac{(2-z^{2k}-z^{-2k})i}{(z-z^{-1})\kappa}
\]
with respect to a subnormalized probability distribution.
Therefore the claim follows.
\end{proof}

Choosing $b' = \lceil \kappa \ln(8/\epsilon) \rceil$,
and feeding a real-on-circle anti-reciprocal Laurent polynomial $f(z)/\ln(8/\epsilon)$,
which is at most 1 in magnitude by the last claim of \cref{lem:approxinverse},
into our algorithm,
we obtain a desired eigenvalue inversion quantum algorithm.
The success probability can be as small as $\Omega(1/(\kappa \log(1/\epsilon))^2)$,
and hence we had better amplify the amplitude for post-selection,
enlarging the quantum gate complexity by a factor of $\calO(\kappa \log (1/\epsilon))$.
Overall, the quantum gate complexity is proportional to the product of
the degree of the Laurent polynomial above and the number of iterations for the amplitude amplification.

\section{Discussion}

We have determined the scope of $SU(2)$-valued periodic polynomial functions and their decomposition (\cref{thm:composition}),
and analyzed the algorithmic aspects.
Our algorithm for the decomposition is not numerically stable in a usual sense
---
a numerically stable algorithm should only require $\polylog(N/\epsilon)$ bits of precision,
rather than $\poly(N\log(1/\epsilon))$.
The instability appears to be unavoidable 
in any method that reduces polynomial degree iteratively by one at a time
 (such as our Step~5),
at least in the early stage of the polynomial degree reduction.
The numerical error arises due to the small norm of leading coefficients in our algorithm,
and the small leading coefficients of an input polynomial are necessary
if a nonpolynomial function admits converging polynomial approximations.
However, it might be the case that the leading coefficient matrix $C^{(m)}_m$ becomes large in norm
rather quickly in Step~5 of our algorithm,
in which case our analysis could be loose.
For a schematic example, 
consider an identity $t^{2n} P + t^{-2n}(I-P) = \left( t P + t^{-1} (I-P) \right)^{2n}$
for any projector $P$.
The numerical instability to decompose the left-hand side into the right-hand side is
far less severe than that in the worst case of Step~5,
and similar situations might occur during the execution of Step~5 for some class of input functions.
This deserves further investigation.

\begin{acknowledgments}
I thank Guang Hao Low for valuable discussions, and Robin Kothari for useful comments on the manuscript.
\end{acknowledgments}

\appendix

\section{Jordan's Lemma and block encoding of Hamiltonians} \label{app:jordan}

The following is a well-known fact, but we include it here for completeness.

\begin{lemma}
Let $P$ and $Q$ be arbitrary self-adjoint projectors on a finite dimensional complex vector space $V$.
Then, $V$ decomposes into orthogonal subspaces $V_j$ invariant under $P$ and $Q$,
where each $V_j$ has dimension $1$ or $2$.
\end{lemma}

For any hermitian operator $H$, let $\supp(H)$ denote the subspace support of $H$, 
i.e., the orthogonal complement of the kernel of $H$.
Clearly, $\supp(H)$ is an invariant subspace of $H$,
and $H$ is invertible within $\supp(H)$.

\begin{proof}
We will find a subspace $W$ of dimension at most 2 that is invariant under both $P$ and $Q$.
This is sufficient since the orthogonal complement of $W$ 
is also invariant and the proof will be completed by induction in the dimension of $V$.

Put $P ' = I - P$ and $Q' = I - Q$, and consider the identities
\begin{align}
&PQP + P Q' P + P' Q P' + P' Q' P' = I,\\
&\supp(PQP) + \supp(P Q' P) + \supp(P' Q P') + \supp(P' Q' P') = V,
\end{align}
where the second equality is 
because the intersection of the orthogonal complements of the four supports is zero.
Therefore, at least one of the four supports is nonzero, 
and without loss of generality assume $S = \supp(PQP) \neq 0$.
Let $\ket \psi \in S$ be an eigenvector of $PQP$; $PQP \ket \psi = a \ket \psi$.
The associated eigenvalue $a$ is nonzero by definition of $S$.
Now consider $W = \mathrm{span} \{\ket \psi, Q \ket \psi\}$.
Observe that $a \ket \psi = PQP \ket \psi = PPQP \ket \psi = a P \ket \psi$,
and hence $P \ket \psi = \ket \psi$. Moreover, $P Q \ket \psi = PQP \ket \psi = a \ket \psi$.
Therefore, $W$ is a nonzero invariant subspace under both $P$ and $Q$.
\end{proof}

This Jordan's lemma can be applied to two hermitian unitaries (reflections) $U_1,U_2$
as any hermitian unitary is $2P - I$ for some projector $P$.
It immediately follows that there is a basis where $U_1$ is diagonal and
$U_1 U_2$ is block-diagonal with at most two-dimensional blocks
and each block belongs to $U(1)$ or $U(2)$.
In any such irreducible two-dimensional block,
both unitaries cannot be scalar multiplications because of the irreducibility and
hence we have
\begin{align}
U_1 = \begin{pmatrix} 1 & 0 \\ 0 & -1 \end{pmatrix}, \quad 
U_2 = \begin{pmatrix} 
    \lambda & e^{i\phi} \sqrt{1-\lambda^2} \\ 
    e^{-i\phi} \sqrt{1-\lambda^2} & -\lambda
\end{pmatrix}
\end{align}
for some real number $\lambda \in [-1,1]$ and an angle $\phi \in \mathbb R$,
up to a permutation of rows and columns.
Therefore, the product $W = -i U_1 U_2$ 
is a rotation in a two-dimensional subspace
that appears in Grover search algorithm~\cite{Grover1996},
and has eigenvalues $\pm e^{\mp i \theta}$ where
$ \sin \theta = \lambda $.
This is relevant in a Hamiltonian simulation problem where
the Hamiltonian is ``block-encoded'' as 
$H = \left( \bra G  \otimes I \right) U_2 \left( \ket G \otimes I\right)$,
which is the case if $H$ is represented as, e.g., a linear combination of Pauli operators.
Here $\ket G$ is a state on a \emph{proper} tensor factor of the Hilbert space on which $U_{2}$ act,
and $U_1 = (2 \ket G \bra G - I)\otimes I$ .

\bibliographystyle{apsrev4-1}
\nocite{apsrev41Control}
\bibliography{qsp-ref}
\end{document}